\def\year{2018}\relax
\documentclass[letterpaper]{article} 
\usepackage{aaai19}  
\usepackage{times}  
\usepackage{helvet}  
\usepackage{courier}  
\usepackage{url}  
\usepackage{graphicx}  
\usepackage{amsmath,amsfonts,amsthm,amssymb,mathrsfs} 
\usepackage[fleqn]{mathtools}
\usepackage[english]{babel}

\def \R{\mathbb{R}}

\def \O{\mathscr{O}}

\def \U{\mathcal{U}}
\def \L{\mathcal{L}}

\def \K{\mathcal{K}}
\def \A{\mathcal{A}}
\def \V{\mathcal{V}}
\def \I{\mathcal{I}}
\def \C{\mathscr{C}}
\def \P{\mathscr{P}}

\def \OO{D}

\def \<{\langle}
\def \>{\rangle}
\newcommand{\QP}{\mathit{QP}}
\renewcommand{\iff}{\Leftrightarrow}
\newcommand{\dimp}{\Leftrightarrow}
\renewcommand{\implies}{\Rightarrow}
\newcommand{\rimp}{\Rightarrow}
\newcommand{\commentout}[1]{}
\newcommand{\journal}[1]{}
\renewcommand{\citeyear}{\shortcite}
\newcommand{\Lbc}{\mathcal{L}^{\mathit{bc}}}
\newcommand{\End}{\mathit{End}}

\def \E{K}

\def \W{\Omega}
\def \w{\omega}

\def \prop{\textsc{prop}}
\def \con{\textsc{con}}

\renewcommand{\phi}{\varphi}
\renewcommand{\blacksquare}{\vrule height7pt width4pt depth1pt}
\newcommand{\shortv}[1]{#1}
\newcommand{\fullv}[1]{#1
}

\newtheorem{thm}{Theorem}[section]
\newtheorem{ex}{Example}[section]
\newtheorem{proposition}{Proposition}[section]
\newtheorem{definition}{Definition}[section]
\newtheorem{lemma}{Lemma}[section]
\DeclareMathOperator*{\argmax}{arg\,max}

\begin{document}
%
\title{Partial Awareness}
\author{  and and Evan Piermont \\
}
\author{Joseph Y. Halpern\\
Computer Science Department \\
Cornell University \\
halpern@cs.cornell.edu
\And
Evan Piermont\\
Economics Department\\
Royal Holloway, University of London\\
evan.piermont@rhul.ac.uk}

\nocopyright
\maketitle
\begin{abstract}
We develop a modal logic to capture partial awareness. 
The logic has three building blocks: objects, properties, and
concepts.  Properties are unary predicates on objects;
concepts are Boolean combinations of properties. 
We take an agent to be partially aware of a concept if she is aware 
of the concept without being aware
of the properties that define it. The logic allows for 
quantification over objects and properties, so that the agent
can reason about her own unawareness. 
%
We then apply the logic to contracts, which we view as
syntactic objects that dictate outcomes based on the
truth of formulas.
We show that when agents are
unaware of some relevant properties, referencing concepts that agents
are only partially aware of can 
improve welfare. 
\end{abstract}

\section{Introduction}
Standard models of epistemic logic assume that agents are
\emph{logically omniscient}: they know all valid formulas and
logical consequences of their knowledge.  There have been many attempts
to find models of knowledge that do not satisfy logical omniscience.
One of the most common approaches involves \emph{awareness}.  Roughly
speaking, an agent $i$ cannot know a valid formula $\phi$ if $i$ is
unaware of $\phi$.  For example, an agent cannot know that either
quantum computers are faster than conventional computers or they are
not if she is not aware of the notion of quantum computer.

There have been many attempts to capture unawareness in the
computer science, economics, and philosophy literature, ranging from
syntactic approaches \cite{FH}, to semantic approaches involving
lattices \cite{HMS03}, to identifying the lack of awareness of $\phi$ with
an agent neither knowing $\phi$ nor knowing that she does not know
$\phi$ \cite{MR94,MR99}.  Most of the attempts involved propositional
(modal) logics, although there are papers that use first-order
quantification as well \cite{BC09,Sil08}.  However, none of these
approaches are rich enough to capture what we will call \emph{partial
unawareness}.   

Perhaps the most common interpretation of lack of awareness
identifies the lack of awareness of $\phi$ with the sentiment
``$\phi$ is not on my radar screen''.
With this interpretation, partial awareness becomes ``some aspects of
$\phi$ are on my radar screen''.%
\footnote{Lack of awareness of $\phi$ has also been identified with
  the inability of compute whether $\phi$ is true (due to
  computational limitations). 
  We do not consider this interpretation in this paper, but partial
  awareness makes sense for it as well---now partial awareness becomes
  ``I can compute whether some aspects of $\phi$ are true.''}
Consider an agent who is in the market for a
new computer.  She 
might be completely unaware of quantum computers, never having heard
of one at all.
Such an agent cannot reason about her value for having a
quantum computer, nor think about for which tasks a quantum computer
would be useful.  But this is an extreme case.
A slightly more
aware agent might be aware of (the concept of) quantum computers, having
read a magazine article about them. She might understand
some properties of quantum computers, for example, that they can factor
integers faster than a conventional computer, but be unaware of 
the notion of qubit state on which quantum computing is based. Such an
agent may well be able to reason about 
her value of a quantum computer despite her less then full awareness.

To capture such partial awareness more formally, we consider a logic
with three building blocks: \emph{objects}, \emph{properties}, and
\emph{concepts}.  
We take a property to be a unary predicate, so it denotes a subset of
objects in the domain;%
\footnote{We could easily extend our approach to allow arbitrary $k$-ary
  predicates, but this would complicate the presentation.  Restricting
  to unary predicates allows us to focus on the more interesting new
  conceptual issues.}  a concept is a Boolean combination of properties.
In each state (possible world), each
agent is aware of a subset of objects, properties, and concepts.

The use of concepts in the context of  awareness, which is (to the
best of our knowledge) original to this paper, is critical in our
approach, and is how we capture \emph{partial} awareness. For a simple example of how we use it, 
suppose that a quantum computer ($Q$) is defined as a
computer ($C$) that possesses an additional
``quantum property'' $\QP$.  That is, $Q$ is defined to be $C \land \QP$
(more precisely, we will have $\forall x (Q(x) \iff C(x) \land
\QP(x))$ as a domain axiom).
A ``partially aware'' agent might be aware
of the concept of a quantum computer but unaware of the specific
Boolean combination of properties that characterizes it.
Contrast this to the cases where the agent is fully unaware (she is unaware of even the concept of a quantum computer) or fully aware (she is aware of both the concept of a quantum computer and also what it means to be one, i.e., the properties $C$ and $QP$).

Once we have awareness in the language, we need to consider 
what an agent knows about her own awareness (and lack of it).  
This is critical in order 
to capture many interesting economic behaviors.
For example, an agent might know (or at least
believe it is possible) that there are aspects of a quantum computer
about which she is unaware
(in our example, this happens to be $\QP$).
In the spirit of Halpern and R\^{e}go \citeyear{HR05b,HR09} (HR from
now on), we
capture this using
quantification over properties: $K_i(\exists P \forall x (Q(x) \iff
C(x) \land 
P(x)))$; although $i$  is unaware of \emph{how} quantum
computers differ from conventional computers, she knows that there is
a distinction that is captured by some property $P$.  

While the agent is unaware of the property
$\QP$,
so cannot reason explicitly about it, 
she is aware that there is some property that relates $C$ and
$Q$. This allows her to 
reason at a sophisticated level about $\QP$.
For example, 
for an arbitrary property $R$,
the statement $K_i\big(\forall P \forall x ((Q(x) \iff C(x)
\land P(x) 
)\implies (P(x) \implies R(x) ) \big)$, combined with the
definition of quantum computer, implies
that the agent knows that $\QP$ implies $R$, even though she is
unaware of $\QP$.
Despite her (partial) unawareness of the notion of quantum computer, 
the agent can reach some substantive conclusions.

With unawareness,
an agent may in general be uncertain about the
relation between $C$ and $Q$; for example, she might also envision a
state where a quantum computer is a computer that satisfies one of two
properties, either $\QP$ or $\QP'$,
but
cannot articulate statements that 
distinguish $\QP$ from $\QP'$. So if the agent wishes to purchase a
$\QP$-computer but not a $\QP'$-computer, she cannot do so.  
This lack of awareness
has important consequences in market settings.

Consider a seller of quantum computers who is fully aware and 
has the ability to teach the buyer; specifically, he can expand the buyer's
awareness, allowing her to discriminate between $QP$ and $QP'$
computers. 
%
%
It is instructive to compare the case of unawareness with the more
standard case of uncertainty 
(with full awareness) in this setting.
In environments of pure
uncertainty, the buyer and seller
are assumed to have a common (and fully understood) space of uncertainty, where each
state fully resolves all payoff-relevant uncertainty.
That is, although they may have uncertainty, both the buyer and
seller understand exactly what information is required in order to resolve the
uncertainty.  A state contains all the relevant information, so, given
a state, both the buyer and seller can (at least in principle) place a
price on all the relevant options.

Suppose that we model the example above using 
two states: $s$, in which $c$ is a
$\QP$-computer, and $s'$, in which $c$ is $\QP'$-computer. If the buyer
knows that the seller knows the state, and the seller can reveal his information in a credible way,
then we can assume without loss of generality that he will always do
so and a transaction will take 
place only in state $s$. To see why, note that in state $s$,
the seller's  
dominant strategy is to reveal his information, ensuring a sale. 
Because 
the buyer knows that the seller knows the state, she will interpret no
information as 
a signal that the state is $s'$.  Thus, in either case the state is
revealed.\footnote{This argument does not rely on the fact that there
  are only two states.
  Suppose that there are $n$ states, say $s_1, \ldots, s_n$,
  and the buyer is willing to pay $p_i$ for the computer if the true
  state is $s_i$,  
   with $p_1 \ge p_2 \ge \ldots \ge p_n$.
   Assume that the seller is willing to sell at any positive
  price.  An easy induction on $k$ shows that, without loss
  of generality, if the true state is $s_k$, the seller might as well
  reveal this fact, as long as the buyer puts positive probability on
  a state $k' < k$.
  (A formalization of this argument also requires common knowledge of
  rationality, or, more precisely, sufficiently deep knowledge of
  rationality.)  
}


If the buyer does not know whether the seller
knows the state, or if the seller cannot credibly reveal the state,
the argument above fails; receiving no information could plausibly be
the consequence of an uninformed seller. If and when information is
revealed or a transaction takes place now depends on the beliefs of 
the agents.
However, an enforceable contract 
can remedy the situation: a contract that stipulates the
sale of the computer
conditional on the true state being $s$ results in essentially the
same 
outcome as the case where the
buyer knows that the seller
knows the true state of the world:
the buyer ends up with the computer if and only if the state is $s$. 
The fact that the
buyer and seller agree on the underlying state space (i.e., the set of
possible states of the world)
makes it possible for an enforceable contract to overcome information
asymmetries.   

We now turn to the situation with unawareness.
If the buyer knows that the seller is aware of all the
intricacies of his product, then 
an analogous argument to the one above shows that we get the same outcome as in
the case of uncertainty.
That is, if the buyer believes the seller is himself aware of the
distinction between $\QP$ and $\QP'$, then she could decide to purchase
a computer $c$ only if he teaches her about the properties and if $\QP(c)$ is true.

On the other hand, if the buyer does not know what the seller is aware
of, then we get a significant divergence between the situation with
awareness and uncertainty.
With unawareness, the seller will again not volunteer information, but
the buyer \emph{cannot} 
draw up a contract  
guaranteeing her the product she wants, since she
cannot articulate the difference between the states.
The efficacy of contracts relies
critically on the parties' common knowledge of the state space,
which in general does not hold in the presence of unawareness.

Note the critical role of the ``partialness'' of awareness here.
If the buyer were fully aware of
the concept, in the sense of her being aware of the
properties that define it, she could write the relevant contracts and
we would have a  case of pure uncertainty. On the other hand, if she
were completely unaware of quantum computers,
she would not be able to reason
about her value, or even consider buying one.


\fullv{
The introduction of concepts allows us to consider agents with
different levels of awareness. For example, perhaps a buyer becomes
aware of a particular company that is offering a commercial quantum
computer understood to be characterized by the concept $T$. If the buyer
knows that $T$ is such that $\forall x\big( T(x) \iff C(x) \land
QP'(x)\big)$, then the buyer, without being explicitly aware of $QP$
or $QP'$, can nonetheless articulate her desire to purchase a $QP$ but
not a $QP'$ computer. Indeed, the buyer can write a contract that
gives her the right to return the computer $c$ in the event that
$T(c)$ is true. In other words, the concept $T$ acts as a proxy for
the property $QP$, allowing the buyer to circumvent her scant
awareness.

Note that the observations above help to explain the prevalence
of costly litigation and contractual disputes. In the world of pure
uncertainty, it can be shown that contracts are always upheld in
equilibrium.  Indeed, if 
uncertainty resolved in such a way as to make some party renege, then
this could be foreseen, and could be addressed by an appropriate contract,
avoiding
costly litigation. However, when parties are aware of different concepts,
optimal complete contracts cannot be drawn up, setting up a barrier to
efficient trade. A legal system that punishes the strategic
concealment of information can help to facilitate trade, as it
provides recourse for unaware buyers who get swindled.
}

\section{A logic of partial awareness}
In this section, we introduce our logic of partial awareness.

\subsection{Syntax}

The syntax of our logic has the following building blocks:
\begin{itemize}
\item A countable set $\O$ of
constant symbols, representing objects.
Following Levesque \citeyear{Lev5}, we assume that $\O$ consists of 
a nonempty set of \emph{standard names} $d_1, d_2, \ldots$, which
may be finite or countably infinite.%
\footnote{Levesque required there to be infinitely many standard names.}
Intuitively, the standard names will represent the domain elements.
We explain the need for these shortly.
\item A countably infinite set $\V^{\O}$ 
  of object variables, which range over objects.
\item A countable set $\P$ of 
unary predicate symbols.
\item    A countably infinite set $\V^{\P}$ of predicate variables.
\item A countable set $\C$ of concept symbols.
\end{itemize}

%
If $d \in \O$, $x \in \V^{\O}$, $P \in
\P$, $Y \in  \V^{\P}$, and $C \in \C$, then $P(d)$,  $P(x)$, $Y(d)$,
$Y(x)$, $C(d)$, and $C(x)$  are
\emph{atomic formulas}.
\journal{they say that the object 
represented by $d$  (resp., $x$) is an instance of the property
represented by $P$ (resp., the property represented by $Y$, the
concept represent by $C$).}
%
Starting with these atomic formulas,
we construct the set of all formulas recursively:
As usual, the set of formulas is closed under conjunction and
negation, so if $\phi$ and $\psi$ are formulas, then so are
$\neg \phi$ and $\phi \land \psi$.
We allow quantification over objects and over unary predicates, so
that if $\phi$ is a formula,  $x \in \V^{\O}$, and $Y \in \V^{\P}$, then
$\forall x \phi$ and $\forall Y \phi$ are formulas.  Finally,
we have two families of modal operators: taking $\{1 \ldots n\}$
to denote the set of agents, we have modal operators $A_1, \ldots,
A_n$ and $K_1, \ldots, K_n$, 
representing awareness and (explicit) knowledge,
respectively.
Thus, if $\phi$ is a formula, then so is $A_i \phi$ and $K_i \phi$.
Let $\L(\O,\P,\C)$ denote the 
resulting language.
A formula that contains no free variables is called a \emph{sentence}. 

\subsection{Semantics}

A model over the language  $\L(\O,\P,\C)$ has to give meaning to each
of the syntactic elements in the language.  We use the standard
possible-worlds semantics of knowledge.  Thus, a model includes 
a set $\Omega$ of possible \emph{states} or \emph{worlds} (we use the two words
interchangeably) and, for each agent $i$, a binary relation
$\K_i$ on worlds.  The intuition is that $(\omega, \omega') \in \K_i$
(sometimes denoted $\omega' \in \K_i(\omega)$) if, in world $\omega$,
agent $i$ considers $\omega'$ possible.

Following HR,
we assume that each state $\omega$ is associated with a language.
Formally, there is a function $\Phi$ on states such that $\Phi(\omega)
= (\O_\w, \P_\w, \C_\w)$, where 
$\O_\w = \O$, $\P_\w \subseteq \P$, and $\C_\w \subseteq \C$.
We discuss the reason for associating a language with each state 
below.  Let $\L(\Phi(\w))$ denote the language associated with state $\w$.
We also assume that associated with each state $\omega$
and agent $i$, there is the set of constant, predicate, and concept symbols
that the agent is aware of; this is given by the function $\A$.
At state $\w$, each agent can only be aware of symbols that are in
$\Phi(\w)$.   
Thus,
$\A_i(\w) \subseteq \Phi(\w)$.   We assume that
all agents are aware of the standard names at every state, so that
$\A(\w)$ includes $\O$.  
Like Levesque \citeyear{Lev5}, we take the domain $D$ of a model over
$\L(\O,\P,\C)$ to consist of the standard 
names in $\O$.
An interpretation $I$
assigns meaning to the constant and predicate symbols in each state;
more precisely, for each state $\w$, we have a function $I_{\w}$ 
taking $\O$ to elements of the domain $\OO$,
$\P$ to subsets of $\OO$,
and $\C$ to Boolean combinations of properties (i.e., predicates).
This last item requires some explanation.  Although elements of $\O$
and $\P$ are mapped to semantic objects (elements in the domain and
sets of elements in the domain, respectively), elements of $\C$ are
mapped to \emph{syntactic} objects: Boolean combinations of
properties.  Let $\Lbc$ denote the Boolean combination of properties;
if $\P' \subseteq \P$, let $\Lbc(\P')$ denote the Boolean combination
of properties in $\P'$.  We require that
$\I_\w(C) \in \Lbc(\Phi(\w))$, so that the Boolean combination
defining $C$ in state $\w$ must be expressible 
in $\L(\Phi(\w))$, the language of $\w$. 
%
We sometimes write $c^I_\w$ rather
than $I_{\w}(c)$, $P^I_\w$ rather than $I_{\w}(P)$,
and $C^I_\w$ rather than $I_{\w}(C)$.
We assume that standard names are mapped to themselves, so that
$(d_i)^I_w = d_i$.

Putting this together, a model for partial awareness has the form
$$M = ( \W, \OO, \Phi, \A_1 \ldots, \A_n, \K_1, \ldots, \K_n, I).$$

The truth of a sentence $\phi \in \L(\O,\P,\C)$ at a state $\w$ 
in $M$
is defined recursively as follows. 
%
\begin{itemize}
\item $(M,\w) \models P(d)$  
 iff $P(d) \in \L(\Phi(\w))$ and $d \in P^I_\w$,
\item $(M,\w) \models \neg \phi$
iff $\phi \in \L(\Phi(\w))$ and $(M,\w) \not\models \phi$,
\item $(M,\w) \models (\phi \wedge \psi)$ 
iff $(M,\w) \models \phi$  and $(M,\w) \models \psi$,
\item $(M,\w) \models C(d)$
iff $C(d)\in \L(\Phi(\w))$ and $(M,\w) \models C^I_\w(d)$,
\item $(M,\w) \models \forall x \phi$
  iff $(M,\w) \models \phi[x/d]$  for all constant symbols $d \in
  \O$,
   where $\phi[x/d]$ denotes the result of replacing all
  free occurrences of $x$ in $\phi$ by $d$,
\item $(M,\w) \models \forall Y \phi$
iff $(M,\w) \models \phi[Y/\psi]$, where   $\psi \in
\Lbc(\Phi(\w))$,%
\footnote{There is an abuse of notation here.  For example,
if $\psi$ is $P \land (Q \lor R)$ and $\phi$ is $Y(d)$,
then $\phi[Y/\psi]$ is $P(d) \land (Q(d) \lor R(d))$; that is, we
apply the arguments of $\phi$ to all predicates in the $\psi$.  We
hope that the intended formula is clear in all cases.}
\item $(M,\w) \models A_i\phi$
 iff $\phi \in \L(\A_i(\w))$,
\item $(M,\w) \models \E_i \phi$
iff $(M,\w) \models A_i\phi$  and $(M,\w') \models \phi$
for all $\w' \in \K_i(\w)$.
\end{itemize}

Note that what we are calling knowledge here is what has been called
\emph{explicit knowledge} in earlier work \cite{FH,HR05b,HR09}: for
agent $i$ to know a formula $\phi$, $i$ must also be aware of it.
Traditionally, $K_i$ has been reserved for \emph{implicit} knowledge
(where no awareness has been required), and $X_i$ has been used to
denote explicit knowledge (where $X_i \phi$ is defined as $K_i \phi
\land A_i \phi$).  Since we do not use implicit knowledge in this
paper, we have decided to use the more mnemonic $K_i$ for knowledge,
even though it represents explicit knowledge.

\fullv{
For the remainder of the paper, we restrict to models where agents
know what they are aware of and 
knowledge essentially satisfies what are called the S5 properties.
Specifically, we restrict to models $M$ where each $\K_i$ is an
equivalence relation and if $\w \in \K_i(\w')$, then $\A_i(\omega) =
\A_i(\omega')$.  This implies that $\A_i(\w) \subseteq \L(\Phi(\w'))$
and $\A_i(\w') \subseteq \L(\Phi(\w))$.
Since $\K_i$ is an equivalence relation, it partitions the states in
$\W$.  
Thus, we can define $\K_i$ by describing the partition.
(In the economics literature, this partition is called $i$'s
\emph{information partition}.)}
\shortv{
For the remainder of the paper, we assume (as is
standard in the literature) that agents know what they are aware of,
so that if $\w \in \K_i(\w')$, then $\A_i(\omega) = \A_i(\omega')$,
and that each $\K_i$ is an equivalence relation, and thus partitions
the states in $\W$.  
We can define $\K_i$ by describing this partition, called 
$i$'s \emph{information partition} in the economics literature.}

Given these assumptions, we can now explain why we need
different languages at different states.  Consider an agent who
considers it possible that she is aware of the whole language.
Thus, 
the agent considers possible a state $\w'$ such that $\A_i(\w') =
\Phi(\w')$.  If we used the same language at all states,
because the agent knows what she is aware of, that would mean that, at
all states that the agent considered possible, she would know the
whole language.  Thus, if an agent is aware of all formulas, then she
would know that she is aware of all formulas.    This is a rather
unreasonable property of awareness.  It was precisely to avoid this
property that Halpern and Rego \citeyear{HR09} allowed different
languages to be associated with different states.

We can also explain our use of standard names.
Note that to give semantics to $\forall x \phi$ and $\forall Y \phi$,
we do syntactic replacements. In the case of $\forall x \phi$, we
consider all ways of replacing $x$ by a constant; in the case of
$\forall Y \phi$, we consider all ways of replacing $Y$ by a
Boolean combination of properties.
This is critical because we define awareness syntactically.
Consider a formula such as 
$\forall Y A_i(Y(d))$.  
The standard approach to giving semantics to such a quantified formula
would say, roughly speaking, that $(M,\w) \models \forall Y A_i(Y(d))$ if
$(M,\w) \models A_i(Y(d))$ no matter which set of objects $Y$ represents. 
The ``no matter which property (i.e., set of objects) $Y$ represents''
would typically 
be captured by including a valuation $V$ on the left-hand side of
$\models$, where $V$ interprets $Y$ as a set of objects; we would then
consider all valuations $V$ that agree on their interpretations of all
predicate variables but $Y$.
Since we treat awareness syntactically, this approach will not work.
We need to replace $Y$ by a syntactic object and then evaluate whether
agent $i$ is aware of the resulting formula.  This is exactly what we do:
$(M,\w) \models A_i(Y(d))$ if 
$(M,\w) \models A_i(\psi(d))$ for $\psi \in \Lbc(\Phi(\w))$.

A sentence $\phi$ is \emph{satisfiable} if there exists a model 
$M$ and a state $\w$ in $M$ such that
$(M,\w) \models \phi$. Given a model $M$, $\phi$ is
\emph{valid} in $M$, denoted $M \models \phi$, if $(M,\w) \models
\phi$ for all $\w \in \W$ such that $\phi \in \L(\Phi(\w))$. Likewise,
for some class of models 
$\mathcal{N}$, $\phi$ is \emph{valid} in $\mathcal{N}$, denoted 
$\mathcal{N} \models \phi$, if $N \models \phi$ for all $N \in
\mathcal{N}$.
Note that when we consider the validity of $\phi$, we
follow Halpern and R\^ego \citeyear{HR09} in requiring only
that $\phi$ be true in states $\w$ such that $\phi \in \L(\Phi(\w))$.
Thus, $\phi$ is valid if $\phi$ is true in all states $\w$ where
$\phi$ is part of the language of $\w$; we are not interested in
whether $\phi$ is true if $\phi$ is not in the language (indeed,
$\phi$ is guaranteed not to be true in this case).

This completes our description of the syntax and semantics.  Our
language is quite expressive.  Among other things, we can faithfully
embed in it the propositional approach considered by HR and the
object-based awareness approach of Board and Chung \citeyear{BC09}.
In more detail, HR consider a propositional logic of knowledge and
awareness, which allows existential quantification over propositions,
so has formulas of the form $\exists X (A_i(X) \land \neg A_j(X))$
(there is a formula that agent $i$ is aware of that agent $j$ is not
aware of).   
We can capture the HR language by replacing each primitive proposition
$p$ by the atomic formula $P(d)$, for some fixed standard name $d$, and
replacing each proposition variable $X$ by $X(d)$.  
This replacement allows us to convert a formula $\phi$ in the HR language to
a sentence $\phi^r$ in our language (the HR language has no analogue of
concepts). 
We can then convert an HR model $M$ to a model $M^r$
in our framework by using the same set of states, taking $\O = \{d\}$,
and taking $\C = \emptyset$, 
so that there are no concepts and a single object. 
It is easy to see that $(M,\w) \models \phi$ iff
$(M^r,\w) \models \phi^r$.
We can also accommodate the object-based awareness models of Board and
Chung \citeyear{BC09}. Here the construction is more
straightforward: a model where $\C = \emptyset$ and $\P_\w = \P$ for
all states $\w$
will do the trick.

\section{Utility under introspective unawareness}\label{sec:utilityassumptions}

In order to explore how the appeal to concepts might be valuable in a
contracting environment, we must add a bit of structure to our
problem, dictating the agents' preferences when they are
unaware. This is more complicated than in previous work because we
have unawareness of properties.
\commentout{
To allow us to focus on the novel features of this paper, we assume in the
remainder of the paper 
that all agents are aware of all 
objects, but may be unaware of properties and concepts. (For a
discussion of how
unawareness of objects affects contracting environments, see
\cite{PE17}.) 
}
Thus, unless the agent knows that she is
aware of all properties, it is always possible there exists
some property that she 
is unaware of. 
We focus on a setting that makes sense for
contracting: namely, one where an agent's utility is determined by
which set of objects he ends up with.  We further simplify things by
assuming that the utility of a set of objects is separable, so it is
the sum of the utility of the individual objects.
Thus, there is no complementarity or substitutability (e.g., it is
not the case that the objects are a left shoe and a right shoe, so
that having one shoe is useless, while having both has high utility);
each agent's preferences can be characterized by a utility
function defined on objects.

In this section, we consider various assumptions on the utility
function, which, roughly speaking, 
correspond to different ways of saying that all that an agent cares about
are the properties and concepts in the agent's language that the
objects satisfy.  In the next section, we consider the consequences of
these assumptions on a contracting scenario.

\mbox{Fix a model $M = ( \W, \OO, \Phi, \A_1 \ldots, \A_n, \K_1, \ldots,
  \K_n, I)$}
  over a language $\L(\O,\P,\C)$.
Let $\U = 
\{U_{i,\w} : \OO \to \mathbb{R}\}_{i \in \{1 \ldots n\}, \w \in \W}$
describe the agents' preferences; specifically, $U_{i,\omega}$
describes agent $i$'s preferences in world $\omega$
by associating with each object its utility (a real number). 
Define, for notational expediency, 
the maps $\prop_\w: D \rightarrow \P$ and
$\con_\w: D \rightarrow \C$ by taking
$\prop_\w(d)
= \{P \in \P_\w : d \in P^I_\w \}$ and $\con_\w(d) = 
\{C \in \C_\w : d \in  C^I_\w\}$. Thus, $\prop_\w$
and $\con_\w$ take a
domain element to the set of properties (resp., concepts) it satisfies
at $\w$. 
Further define  
$\prop^{\A_i}_\w(d) = \prop_\w(d) \cap \A_i(\w) 
$ and $\con^{\A_i}_\w(d) = \con_w(d) \cap
\A_i(\w)$;
thus $\prop^{\A_i}_\w$ and
$\con^{\A_i}_\w$ are the restrictions of $\prop_w$ and $\con_w$ to 
agent $i$'s awareness.

Our assumptions relate the agents' preferences over domain elements to the
properties (and concepts) that these elements
possess. Because we allow for introspection, it is possible
for an agent to care about aspects of an object that she is unaware
of, although she will not be able to articulate exactly why she
has such a preference.
Thus, we take as a starting point the minimal restriction
ensuring an agent's subjective valuation depends only on
properties and concepts that the agent can articulate.

\begin{itemize}
  \item[A1.]
        For all $d, d' \in \OO$ and all states $\w, \w' \in \W$,
if $\prop_\w(d) = \prop_{\w'}(d')$ then $U_{i,\w}(d) = U_{i,\w'}(d')$.
\end{itemize}

A1 can be viewed as the conjunction of
two assumptions:
first, if two different objects $d$ and $d'$ possess exactly the same
properties at a state $\w$, then they are valued identically at $\w$;
second, if a given object $d$ has the same properties at two different
states $\w$ and $\w'$, then $d$ has the same value at both $\w$ and $\w'$.
Since each concept is defined (at a possible state) as a Boolean combination
of properties, if two objects satisfy the same properties then they must
also be instances of the same concepts.
Thus, A1 is equivalent to the assumption that if $\prop_\w(d) =
\prop_{\w'}(d)$ and  $\con_\w(d) =
\con_{\w'}(d')$ then $U_{i,\w}(d) = U_{i,\w'}(d')$.  

Under A1, an agent does not care about the
\emph{label} assigned to an object---if $d$ and $d'$ satisfy the same
properties (and so the same concepts) the agent does not care that $d$ is
called ``$d$" and $d'$ called ``$d'$". A1 also
rules out social preferences, or preferences that depend on epistemic
conditions. For example, A1 does not allow
agent $i$ to value an object $d$ according to
agent $j$'s valuation, or even agent $j$'s current knowledge of $i$'s
valuations, although this might be relevant if $i$ is 
interested in reselling $d$ to $j$.   Finally,
A1 rules out the case 
where an agent values the same properties differently in different
states.

A1 allows
an agent to value $d$ 
and $d'$ differently even if she can express no distinction between
$d$ and $d'$ (conditional on the state.
This can happen if $d$ and $d'$ differ on properties that the agent is
unaware of.
Our next assumption reduces this flexibility in valuations, mandating
that an agent's valuation depends only on aspects of the state of
which she is aware.

\begin{itemize}
 \item[A2.]
If $\prop^{\A_i}_\w(d) = \prop^{\A_i}_{\w'}(d')$ and $\con^{\A_i}_\w(d) = \con^{\A_i}_{\w'}(d')$ then $U_{i,\w}(d) = U_{i,\w'}(d')$.
\end{itemize}


A2 says that the DM's
valuation of objects 
cannot differ unless the objects are distinguished in some way
the agent is aware of. If
two objects are the same in every way that the agent can articulate,
then she assigns them the same value. 
It is consistent with A2 that an agent $i$ ascribes different
utilities to $d$ and $d'$ even though $i$ is not aware of any property
that distinguishes $d$ and $d'$.  This can happen if $d$ and $d'$ satisfy
different concepts.  In that case, $i$ knows that some property
must distinguish $d$ and $d'$, but it is a property that $i$ is not
aware of.
For instance, an agent might value a quantum computer more than a conventional
one even when she does not understand exactly how to define a quantum
computer.  
As this discussion shows, unlike A1, in A2 we must
explicitly refer to concepts; even though concepts are built from
properties, the agent can be aware of concepts that are defined by
properties that the agent is not aware of.
\commentout{
In such a case, the agent understands
that, and can articulate how, $a$ and $b$ are different, despite not
being aware of a specific property that distinguishes them. 
}

\begin{itemize}
  \item[A3.]
If $\prop^{\A_i}_\w(d) = \prop^{\A_i}_{\w'}(d')$ then $U_{i,\w}(d) = U_{i,\w'}(d')$.
\end{itemize}

A3 says that an agent bases her preferences
only on the properties of which she is aware (and not concepts). 
A3 can be thought of as a principle of
neutrality towards  
unawareness; although two objects are distinguishable (say $d$ is an
instance of $C$ while $d'$ is not), the agent values them identically as
long as they are not distinguishable by properties that the agent is
aware of. That is, while the
agent knows there must be a property that separates $d$ and $d'$, because
she is unaware of such a property, so she places no value on it. For
example, while the agent might understand that there is a difference
between classical and quantum computers, because she does not
understand \emph{how} these entities differ, she values them
equally.
It is immediate
that A3 implies both A1 and A2.
Moreover, in a model without unawareness of properties or concepts,
A1, A2, and A3 collapse to
the same restriction.

 While it may at first seem 
unreasonable to base preferences on properties of which you are
unaware (as is allowed by A1), it actually is not so uncommon.  People
may prefer stock 
$d$ to $d'$ although they cannot articulate a reason that $d$ is
better.  Nevertheless, because they see other people buying $d$ and
not $d'$, they assume that there is some
significant property $P$ of $d$ (of which they are not aware) that
$d'$ does not possess that accounts 
for other people's preferences.  
This can be modeled using
the fact that different states have different languages associated
with them.
An agent $i$ might not be aware of $P$ at a world
$\omega$, but considers a world $\omega'$ possible such that $P(d)$
holds and $P(d')$ does not.  Moreover, he considers $P$ a ``good''
property; objects that have property $P$ get a higher utility than
those that do not, all else being equal.  We remark that such
reasoning certainly seems to play a role in the high valuation of some
cryptocurrencies!
Of course, if $i$ has a predicate in the language
that says ``other people like it'', then $d$ and $d'$ might be
distinguishable using that predicate.  This observation emphasizes the
fact that A1 and A2 are very much language-dependent.
If the agent
cannot express various properties in his language, then he may not be
able to make distinctions relevant to preferences.
(See \cite{BHP13}
for an approach to game theory and decision theory that assumes that
utilities are determined by the description of a state in some
language.)

\section{Contracts and conceptual unawareness}\label{sec:contracts}

We now consider the effect of A1, A2, and A3 on 
simple interpersonal contracts. Unlike the bulk of
the Economics literature, where contracts are functions from a state
space to outcomes, we here take a contract to be a \emph{syntactic}
object---it makes direct reference to the language of our logic. Real
world contracts are syntactic, in that they must literally articulate
the contingencies on which they are based. But our motivation for
considering syntactic contracts 
is more than a pursuit of descriptive accuracy. In
models of unawareness, the set of contingencies that can be
contracted on is a direct consequence of what
agents can articulate. So, by considering the language
that the agents use,
we can directly examine the welfare
implications of awareness.%
\footnote{For contracts that do not involve
  concepts, what we are doing could also 
  be done in a purely semantic framework, for example, that of
  Heifetz, Meier, and Schipper \citeyear{HMS03}.
  However, contracts that involve concepts
cannot be expressed in a purely semantic framework, since concepts
represent syntactic objects.  
Agents can be aware of a
concept without being aware of the properties that the concept represents (as
in the case of quantum computers); because of this, concepts allow us
to indirectly get at 
awareness of unawareness. 
Awareness of unawareness seems difficult to
express in a purely 
semantic framework (and, indeed, cannot be expressed in the framework
of Heifetz, Meier, and Schipper).
It seems to us that the use of concepts is  
  indispensable in understanding how unawareness drives novel behavior
    in contracting environments; this is one of our reasons for
    modeling unawareness syntactically. 
  Syntactic contracts were
considered by Piermont \citeyear{Piermont}, with similar motivation.}

Suppose that we have two agents, 1 and 2.  Let 
$M$ be the model that describes their uncertainty and awareness, and
suppose that their preferences are characterized by $\U$.    Assume
that agent $i$ is initially endowed with the set of domain
elements $\End_i$, for 
$i=1,2$, where $\End_1$ and $\End_2$ are disjoint.
For simplicity, we assume (as is standard) that each agent will
consume (i.e., use) exactly one object. 
Without trade, each agent  
can consume only an object from her own endowment; with trade, they
may be able to do better.  Let $\End_1
{\oplus} \End_2$ denote all pairs   $(d_1,d_2)$ of elements in $\End_1 \cup
\End_2$ such that $d_1 \neq d_2$.  
Note that we might have $(d_1,d_2) \in \End_1 \cup \End_2$ even if
$d_1$ and $d_2$ are both in $\End_2$ (and not in $\End_1$); the agents
may both consume something that was in agent 2's initial endowment.

A \emph{contract} is a pair $\< \Lambda, c\>$, where $\Lambda$ is
a finite set of sentences and $c$ is a function from 
$\Lambda$ to $\End_1 {\oplus} \End_2$.  Let $c_i$ denote the
$i^{\mathrm{th}}$ component of $c$; that is, if $c(\lambda) = (d_1, d_2)$,
then $c_i(\lambda) = d_i$.   The intuition is that $c_i$ dictates which
object 
should be consumed by agent $i$, contingent on the truth of the sentences in
$\Lambda$.
Given a model $M$,
contract $\< \Lambda, c\>$ must satisfy 
\begin{enumerate}
\item $M \models \bigvee_{\phi \in \Lambda} \phi$, and 
\item $M \models \neg(\phi \land \psi)$ for all distinct sentences $\phi,\psi
\in \Lambda$. 
\end{enumerate}

The first condition states that some sentence in $\Lambda$ is true in
every state (so the contract is complete), the second that the true
sentence is unique at every state (so the contract is well defined).
A contract is \emph{articulable} in state $\w^*$,
sometimes denoted \emph{$\w^*$-articulable},
if
$\Lambda \subseteq \L(\A_1(\w^*) \cap \A_2(\w^*))$,
that is, if both agents are aware of all the statements in
the contract. (Presumably, the act of reading the contract makes them
aware all the statements even if they weren't aware of them beforehand.%
\footnote{The fact that agents might not be aware of all statements in a
  contract before the contract is written clearly has strategic
  implications.  Agent 1 may prefer to leave a clause out of a
  contract rather than making agent 2 aware of the issue.   This issue
  is studied to some extent by Filiz \citeyear{Filiz07} and Ozbay
\citeyear{Ozbay06}.})
Note that if a contract is articulable in $\w^*$ then $\Lambda \subseteq
\L(\Phi(\w^*))$.

%
By conditions 1 and 2, in each state $\w$, there is
a unique sentence in $\Lambda$ that is true in $\w$: call that sentence
$\phi_\w$.  We take the outcome of the contract in state $\w$ to be
$c(\phi_\w)$.  We abuse notation and write $c(\w) = (c_1(\w),c_2(\w))$
to denote the outcome of the contract in state $\w$.   Thus, the value
of the contract to agent $i$ in state $\w$ is $U_{i,\w}(c_i(\w))$. 

We are interested in the value of concepts 
as a contracting device. To get at this, we must examine the difference
between the optimal (or equilibrium) contract when $\Lambda$ can be
any subset of the language and the optimal (or equilibrium) contract
when $\Lambda$ cannot refer to 
concepts. Given $M$, $\U$, and endowments  $\End_1,\End_2 \subseteq \OO$, 
say that a contract $\< \Lambda,c\>$ is \emph{$\w$-efficient} if
there is no pair of objects $(d_1, d_2) \in \End_1{\oplus}\End_2$ such
that  
$$U_{i,\w}(d_i)) \geq U_{i,\w}(c_i(\w))$$
for $i \in \{1,2\}$, with at least one inequality strict, and
\emph{efficient} if it is $\w$-efficient for all $\w \in \W$. In other
words, a contract is efficient if it realizes all the gains from
trade, so there is no trade that would leave both agents better off.
Finally, a contract is \emph{$\w$-acceptable for agent $i$} if 
$$U_{i,\w'}(c(\w)) \geq \max_{d \in \End_i} U_{i,\w'}(d)$$
for all $\w' \in \K_i(\w)$.
Agent $i$ facing a take-it-or-leave-it offer for an acceptable contract
will prefer that contract to her outside option (i.e.,
consuming an object in $\End_i$).

The following examples illustrate how limited awareness can impede the efficiency of contracts and how the reference to concepts
can help an agent articulate her preference, tempering the
effect of unawareness.   

\begin{ex}\label{ex:quantum1}
{\rm
A buyer (agent 1) is trying to purchase a computer from a firm (agent
2). $\End_1 = \{d^\$\}$ and $\End_2 = \{d^{cmp}\}$, where $d^\$$ is a fixed amount of
money  and $d^{cmp}$ is the computer in 
question. There are three states: $\W = \{\w_1, \w_2, \w_3\}$.  There
are three predicates, $\P = \{P,Q,R\}$.
We have $R^I_{\w_1} = R^I_{\w_2} = R^I_{\w_3} = \{d^\$\}$; in addition,
 $P^I_{\w_1} = P^I_{\w_2} = Q^I_{\w_1} = \{d^{cmp}\}$
and $P^I_{\w_3} = Q^I_{\w_2} = Q^I_{\w_3} = \emptyset$, so that, in the
three states, $d^{cmp}$ has
properties $P$ and $Q$, property $P$, and no properties, respectively.  
%
There is also a single concept, that of a quantum computer, $QC$:
$QC^I_{\w_1} = QC^I_{\w_2} = P \land Q$ and $QC^I_{\w_3} = \neg P
\land \neg Q \land \neg R$. Therefore $c$ is an instance of $QC$ in
states $\w_1$ and $\w_3$. The buyer prefers to purchase the computer
if and only if it has property $Q$; thus, the buyer's utility is such that
$U_{1,\w_1}(d^{cmp}) > U_{1,\w_1}(d^\$)$ and $U_{1,\w_k}(d^\$) > U_{1,\w_k}(d^{cmp})$
for $k \in \{2,3\}$.
Moreover, $U_{i,\w}(d^\$) = U_{i,\w'}(d^\$)$ for all agents $i$ and states
$\w, \w'$.  
The firm wants to sell the computer in all states.  
%
For now, assume that both agents have full awareness, and their
information partitions are given by $\K_1 =
\big\{\{\w_1,\w_2,\w_3\} \big\}$ and $\K_2 =
\big\{\{\w_1,\w_2\},\{\w_3\}\big\}$. 
Since there is no unawareness, this model trivially satisfies
%
A3 (and hence A1 and A2).
Because the buyer does not know the state, she is
unwilling to make any unconditional trade (all constant contracts are
unacceptable for the buyer). However, this is easily remedied by 
the use of a contract. The obvious contract, 
where $\Lambda = \{Q(d^{cmp}), \neg Q(d^{cmp})$ and the function $c$
is given by
\begin{align*}
Q(d^{cmp}) &\mapsto (d^{cmp},d^\$) \\
\neg Q(d^{cmp}) &\mapsto (d^\$,d^{cmp}),
\end{align*}
is clearly efficient and acceptable to all parties.
\blacksquare
}
\end{ex}


Example~\ref{ex:quantum1}  highlights how contracting can facilitate trade in uncertain
environments. Despite the fact that agents do not know which state
has obtained, they can eliminate uncertainty by appealing to
contracts. The next example 
illustrates the issues that arise when awareness is limited.  

\begin{ex}\label{ex:quantum2}
{\rm
Let $M$ and $\U$ be as in Example~\ref{ex:quantum1}, except that
  now that agents are not 
  completely aware.  Specifically, 
$\A_{i}(\w) = (\O,\{P,R\},\C)$ for all $\w \in \W$ and $i \in \{1,2\}$.
Both agents are unaware of $Q$.
This model satisfies assumptions
A1 and A2, but
not A3.  
The contract described in the previous example is no longer articulable.
We can circumvent the agents' linguistic limitations by writing a
contract in terms of the concept $QC$. Indeed, the consider
$(\Lambda,c)$, where 
$\Lambda = \{P(d^{cmp}) \land Q(d^{cmp}), \neg (P(d^{cmp} \land
Q(d^{cmp}))$ and $c$ is given by
\begin{align*}
P(d^{cmp}) \land QC(d^{cmp}) &\mapsto (d^{cmp},d^\$) \\
\neg \big(P(d^{cmp}) \land QC(d^{cmp})\big) &\mapsto (d^\$,d^{cmp}),
\end{align*}
implements the same consumption outcomes as the contract in
Example~\ref{ex:quantum1}. 
}
\blacksquare
\end{ex} 


In Example~\ref{ex:quantum2}, the buyer wants to purchase $d^{cmp}$ only when $Q(d^{cmp})$ is
true.  Since she is unaware of the property $Q$, and
knows only that there is some property (that is a conjunct of $QC$) that is
desirable, 
she cannot \emph{directly} demand a computer with
property $Q$. Before analyzing the contract above,
notice if the buyer knew the true state was not $\w_3$, then
she could get away with the simple contract that demands $d^{cmp}$ whenever
$QC(d^{cmp})$ is true. In states $\w_1$ and $\w_2$, the interpretation of
$QC$ is constant, and, given that $P(d^{cmp})$ is true in both states, $QC(d^{cmp})$ is
equivalent to $Q(d^{cmp})$ in these states, so the buyer could use the
concept of a quantum 
computer as a proxy for the property $Q$. 

This simpler contract is not acceptable when the buyer considers all
three states possible. In state $\w_3$, the interpretation of a
quantum computer is different, so that while $d^{cmp}$ is an instance of
$QC$ in $\w_3$, it does not satisfy $Q$. The buyer is uncertain about the
definition of a quantum computer; while she is unaware of the exact
definition in each state, she can articulate the difference: in some
states, $P$ is a property of quantum computers, while in others it is
not.  By exploiting this difference, she can construct the welfare-optimal
contract---she demands $d^{cmp}$ whenever it possesses the property that
defines a quantum computer in addition to $P$.  

These examples show how the awareness of agents can affect the set
of trading outcomes that can be implemented via (syntactic)
contracts. Collectively, they suggest a connection between the
efficacy of contracting and the relationship between preference and
properties as embodied by assumptions A1--A3.
Under the additional assumption that 
the agents are aware of the same things in the actual world (a
reasonable assumption if we assume that the language talks only about
contract-relevant features, and both agents have read the contract, so
are aware of all the properties and concepts that the contract mentions), this
connection is made formal in the following result, whose proof (like
that of all other theorems) is left to the full paper, which can be
found on arxiv.

\begin{thm}\label{thm:contract1}
Given a model $M = ( \W, \OO, \Phi, \A_1 \ldots, \A_n, \K_1,
\ldots, \K_n, I)$, preferences $\U$, endowments  $\End_1,\End_2 
\subseteq \OO$, and a state $\w^* \in \W$ such that $\A_1(\w^*) =
\A_2(\w^*)$
and $\K_1(\w^*) \cup \K_2(\w^*)$ is finite,
the following hold: 
\begin{enumerate}
\item[(a)] If $\<M,\U\>$ satisfies A1 and A2, then
there exists a contract that is $\w^*$-articulable, $\w^*$-efficient,
  and $\w^*$-acceptable for $i = 1,2$.
 \item[(b)] If in addition, for $i=1,2$, $\A_i$ is constant on $\K_1(\w^*) \cup \K_2(\w^*)$, then
there exists a contract that is $\w^*$-articulable, $\w'$-efficient,
  and $\w'$-acceptable for all 
  $\w' \in \K_1(\w^*) \cup \K_2(\w^*)$.  
\item[(c)] If in addition $\<M,\U\>$ satisfies A3,
then there exists a contract $\<\Lambda, c\>$ that is
$\w^*$-articulable, $\w'$-efficient, 
and $\w'$-acceptable for $i$ at all
$\w' \in \K_1(\w^*) \cup \K_2(\w^*)$ such that   
$\Lambda \subseteq \Lbc$.  
\end{enumerate}

\end{thm}   


Theorem~\ref{thm:contract1}(a) says that if agents' preferences can
depend only on properties and concepts that they are aware of,
then gains from trade can be fully realized. 
Even if agents are unaware of some
preference-relevant properties, as long as they do not strictly prefer
one object to another without being aware of some tangible way that the
objects differ, then they can still
articulate an optimal contract. As Example~\ref{ex:quantum2} shows,
this  contract might need to mention concepts.
Part (b) states that if, in addition, each agent knows what
the other is aware of, then each of them knows
that gains from trade can be achieved.
That is, both agents know that, no matter what the true state of the
world is (from their perspective), trading is worthwhile.
%
Theorem~\ref{thm:contract1}(c) says that if agents' preferences 
depend only on the properties that they are aware of, then they gain
nothing from the ability to contract over concepts; there is a
contract that they know to be efficient that makes reference only to
properties.

Theorem~\ref{thm:contract1} requires agents to be aware of the same
properties in $\w^*$.  As we argued above, this is a reasonable
assumption.  
As seen by the following example, the assumption is also necessary.

\begin{ex}
{\rm
  Let $\End_1 = \{d_1\}$ and $\End_2 = \{d_2\}$.  Consider a language with three
  predicate symbols $P$, $Q$, and $R$.  Let $M$ be a model with three
  states, $\w_1$, $\w_2$, and $\w_3$, where 
  $P^I_{\w_1} = \{d_1\}$, $Q^I_{\w_3} = \{d_2\}$, $P^I_{\w_2} = P^I_{\w_3}
  = Q^I_{\w_1} = Q^I_{\w_2} = \emptyset$, $R^I_\w = \{d_1\}$ for all
  states $\w$, 
  the only information set for both agents is $\{\w_1,\w_2,\w_3\}$ (so
  both agents consider all three worlds possible at all worlds),
$\A_{1}(\w) = (\O,\{P,R\},\emptyset)$, and $\A_{2}(\w) =
(\O,\{Q,R\},\emptyset)$ for each state $\w$. Now consider what happens when
agent 1 wants $d_1$ only when it has property $P$ (so wants to trade in
states $\w_2,\w_3$) and agent 2 wants $d_2$ only when it has property $Q$
(so wants to trade in states $\w_1$ and $\w_2$).  
%
Thus, an efficient contract must induce trade in state $\w_2$ and an
acceptable contract cannot induce trade except in state
$\w_2$. However, neither agent alone can propose such a contract. From
agents 1's perspective, states $\w_2$ and $\w_3$ are
indistinguishable in the sense that all objects that he is aware of
satisfy the same properties in both states, and from 2's perspective,
$\w_1$ and $\w_2$ are 
indistinguishable.  
\blacksquare
}
\end{ex}

\section{Axiomatization and complexity}

We can adapt the axioms used by Halpern and R\^{e}go \citeyear{HR09}
to get a sound and complete axiomatization for our logic,
provided that the set $\P$ of predicates is infinite.
This assumption seems reasonable, given that we are mainly interested
in agents who are never sure that they are aware of all predicates.
(HR make an analogous assumption.)  

Consider the following axiom system, which we call AX.
\smallskip

\noindent{{\bf Axioms:}}
\begin{description}
\item[{\rm Prop.}] All substitution instances of valid formulas of
propositional logic.

\item[{\rm AGP.}] \mbox{$A_i\phi \dimp
 (\land_{P\in \P \cap \Phi(\phi)} A_i P(x)) \land 
  (\land_{C\in \C \cap \Phi(\phi)} A_i C(x))$}, where $x$ is an
  arbitrary object variable  
  and $\Phi(\phi)$ consists of all 
  the predicate and concept symbols in $\P \cup \C$
that appear in $\phi$.%
\footnote{As usual, the empty conjunction is taken to be
vacuously true, so that $A_i \phi$ is vacuously true
if there are no symbols in $\P \cup \C$ occur in $\phi$.}

\item[{\rm KA.}] $A_i \phi \rimp K_i A_i \phi$


\item[{\rm K.}] $(K_i\varphi\land K_i(\varphi\rimp\psi))\rimp
K_i\psi$.

\item[{\rm T.}] $K_i\varphi\rimp \varphi$.

\item[{\rm 4.}] $K_i\varphi\rimp K_iK_i\varphi$.

\item[{\rm 5.}] $(\neg K_i\varphi \land A_i \phi) \rimp K_i\neg K_i\varphi$.

\item[{\rm A0.}] $K_i\varphi\rimp A_i\varphi$.

\item[{\rm Con.}] $\exists X (\forall x (C(x) \iff X(x)))$.
 

\item[{\rm $1_{\forall_x}$}] $\forall x \psi \rimp \psi[x/c]$ for
$c \in \O$. 

\item[{\rm $1_{\forall_X}$.}] $\forall X\varphi\rimp\varphi[X/\psi]$
if $\psi$ is either in $\Lbc$ or a concept.

%

\item[{\rm ${\rm K}_{\forall x}$.}] $\forall x(\varphi\rimp\psi)\rimp (\forall
x \varphi\rimp\forall x\psi)$.

\item[{\rm ${\rm K}_{\forall X}$.}] $\forall X(\varphi\rimp\psi)\rimp (\forall
X \varphi\rimp\forall X\psi)$.

\item[{\rm ${\rm N}_{\forall x}$.}] $\varphi\rimp\forall x\varphi$ if
  $x$ is not free in $\varphi$.

\item[{\rm ${\rm N}_{\forall X}$.}] $\varphi\rimp\forall X\varphi$ if
  $X$ is not free in $\varphi$.

\item[{\rm Barcan$_x$.}] $\forall x K_i\varphi\rimp K_i\forall
  x\varphi$.

\item[{\rm Barcan$_X$.}] $(A_i(\forall X \phi) \land \forall X
  (A_i(X(c))) \rimp K_i\varphi) \rimp K_i(\forall X A_i(X(c)) \rimp
  \forall X  \varphi)$.

\item[{\rm FA$_X$.}] $\forall X \neg A_i (X(c)) \rimp K_i(\forall X
  \neg A_i(X(c)))$. 

\item[{\rm Fin$_x$.}]
  If $\O = \{c_1, \ldots, c_n\}$,
  then $\forall x \phi \dimp 
  \phi[x/c_1] \land \ldots \land \phi[x/c_n]$.


\end{description}

\noindent {\bf Rules of Inference:}

\begin{description}  
\item[{\rm MP.}] {F}rom $\varphi$ and $\varphi\rimp\psi$ infer $\psi$
(modus ponens).

\item[{\rm Gen$_K$.}] {F}rom $\varphi \land A_i \phi$ infer $K_i \varphi$.
\item[{\rm Gen$_{\forall_x}$.}]
From $\phi$ infer $\forall x \varphi[c/x]$, where $c \in \O$.

 \item[{\rm Gen$_{\forall_X}$.}]
If $P$ is a predicate symbol, then from $\phi$ infer $\forall
X \varphi[P/X]$.

\end{description}

\begin{thm}\label{thm:completeness} AX is a sound and complete axiomatization of
  $\L(\O,\P,\C)$ with respect to the class of models of partial
  awareness,
if $\P$ is infinite.
  \end{thm}


Since the logic is axiomatizable, the validity problem is recursively
enumerable.  This is also a lower bound on its complexity, even if we
do not allow quantification over predicates, since
first-order epistemic logic with just two unary predicates 
was shown by Kripke \citeyear{Kripke62} to be undecidable.  Kripke's proof
used the 
well-known fact that first-order logic with a single binary predicate $R$
is undecidable, and the observation that $R(x,y)$ can be represented
as $\neg K \neg (P(x) \land Q(y))$.  (We must add the formula $\forall x (A
(P(x) \land Q(x))$ to ensure that awareness does not cause a problem.)
\fullv{We thus get}
Thus,

\begin{thm}\label{thm:complexity} The validity problem for the
  language $\L(\O,\P,\C)$
  in the class of models of partial awareness is r.e.-complete if
  $|\P| \ge 2$.
  \end{thm}

\section{Conclusion}

We have defined and axiomatized a modal logic
that captures partial unawareness by allowing
an agent to be 
aware of a concept without being aware
of the properties that define it. The logic also allows agents to reason about their own unawareness. 
\commentout{
We apply this logic to interpersonal contracts. The set of feasible
contracts is limited by the formulas the agents are aware of.
We show that when agents are
unaware, referencing concepts that agents are only partially aware of can
improve welfare. }
We show that such a logic is critical for analyzing interpersonal
contracts, and that referencing concepts that agents are only
partially aware of can improve welfare.
We believe that the logic should also be applicable to other domains,
such an analyzing communication between people. We hope to consider
such applications in the future.

We are also interested in analyzing dynamic aspects of awareness, and
applying this to contracts.
Our analysis of contracts assumed that both agents were aware of all
statements in a contract.  This makes sense after the contract has
been signed, but may well not be true before the contract is written.
We believe that an extension of our language to deal with the effects of
making other agents aware of certain formulas will allow us explore
and analyze the dynamic process of contract writing.  

\paragraph{Acknowledgments:}
Halpern was supported in part by NSF 
grants IIS-1703846 and IIS-1718108, 
ARO grant W911NF-17-1-0592, and a grant from the Open Philanthropy
project. 

\bibliography{joe,z}
\bibliographystyle{aaai}


\newenvironment{oldthm}[1]{\par\noindent{\bf Theorem #1:} \em \noindent}{\par}
\newenvironment{oldlem}[1]{\par\noindent{\bf Lemma #1:} \em \noindent}{\par}
\newenvironment{oldcor}[1]{\par\noindent{\bf Corollary #1:} \em \noindent}{\par}
\newenvironment{oldpro}[1]{\par\noindent{\bf Proposition #1:} \em \noindent}{\par}
\newcommand{\othm}[1]{\begin{oldthm}{\ref{#1}}}
\newcommand{\eothm}{\end{oldthm} \smallskip}
\newcommand{\olem}[1]{\begin{oldlem}{\ref{#1}}}
\newcommand{\eolem}{\end{oldlem} \smallskip}
\newcommand{\ocor}[1]{\begin{oldcor}{\ref{#1}}}
\newcommand{\eocor}{\end{oldcor} \smallskip}
\newcommand{\opro}[1]{\begin{oldpro}{\ref{#1}}}
\newcommand{\eopro}{\end{oldpro} \smallskip}

\fullv{
 \section{More on contracts}
In this section, we expand on the discussion of contracts in
Section~\ref{sec:contracts}.  We start with the proof of
Theorem~\ref{thm:contract1}.  For the reader's convenience, we repeat
the statement of the theorem. 

\smallskip

\othm{thm:contract1}
  Given a model $M = ( \W, \OO, \Phi, \A_1 \ldots, \A_n, \K_1,
\ldots, \K_n, I)$, preferences $\U$, endowments  $\End_1,\End_2 
\subseteq \OO$, and a state $\w^* \in \W$ such that $\A_1(\w^*) =
\A_2(\w^*)$
and $\K_1(\w^*) \cup \K_2(\w^*)$ is finite,
the following hold: 
\begin{enumerate}
\item[(a)] If $\<M,\U\>$ satisfies A1 and A2, then
there exists a contract that is $\w^*$-articulable, $\w^*$-efficient, and
  $\w^*$-acceptable for $i = 1,2$. 
 \item[(b)] If in addition, for $i=1,2$, $\A_i$ is constant on
  $\K_1(\w^*) \cup \K_2(\w^*)$, then
there exists a contract that is $\w^*$-articulable, $\w'$-efficient,
  and $\w'$-acceptable for all 
  $\w' \in \K_1(\w^*) \cup \K_2(\w^*)$.  
\item[(c)] If in addition $\<M,\U\>$ satisfies A3,
then there exists a contract $\<\Lambda, c\>$ that is
$\w^*$-articulable, $\w'$-efficient, 
and $\w'$-acceptable for $i$ at all
$\w' \in \K_1(\w^*) \cup \K_2(\w^*)$ such that   
$\Lambda \subseteq \Lbc$.  
\end{enumerate}
\eothm

\begin{proof}
  Define an equivalence relation $\sim$ on states by taking $\w \sim
\w'$ ff $\w$ and $\w'$ agree on all sentences that the agents are aware
of in state $\omega^*$ (recall that the agents are aware of the same
sentences in $\omega^*$); that is, $\w \sim \w'$ if
 \begin{align*}
\{ \phi \in& \L : (M,\w) \models \phi, (M,\w^*) \models A_1\phi \} = \\
&\{ \phi \in \L : (M,\w') \models \phi, (M,\w^*) \models A_1\phi \}.
 \end{align*}
(Note that since $\A_1(\w'^*) = \A_2(\w^*)$ by assumption, we could
 replace either occurrence of $A_1$ above by $A_2$ without affecting $\sim$.)
 For each $\w \in \W$, let $[\w]$ denote the $\w$'s \emph{cell}, that
is, the set of states $\sim$-equivalent to $\w$.

\begin{lemma}
\label{Lemma:implemented}
Let $\kappa: \W \to \End_1{\oplus}\End_2$ be a function such that
$\kappa^{-1}((d_1,d_2))$ is a union of $\sim$-cells for all
$(d_1,d_2) \in 
\End_1 \oplus   \End_2$. Then for all finite subsets
$\W'$ of $\W$,
there exists a
$\w^*$-articulable contract
$\<\Lambda, c\>$ 
such that  
$c(\w) = \kappa(\w)$ for all $\w \in \W'$.
\end{lemma}

\begin{proof}
Fix a finite subset $\W' = \{w_1, \ldots, \w_n\} \subseteq \Omega$.
First suppose that $\w_i \not\sim \w_j$ for $i \ne j$,
so that for all $i,j \leq n$, $i\neq j$, there exists a
sentence $\phi_{ij} \in  \L(\A_1(\w^*) \cap \A_2(\w^*))$ such that
$(M,\w_i) \models \phi_{ij}$ and $(M,\w_j) \not\models \phi_{ij}$.
%
For each $i \leq n$ let $\psi_i
= \bigwedge_{\{j \neq i\}} \phi_{i,j}$.
For $i = 1, \ldots, n$, let $$\lambda_i
= \bigwedge_{j\neq i} \neg \psi_j \land \psi_i; $$
let $\lambda_{n+1} = \land_{i=1}^n \neg \lambda_i$.
Thus,
for $i=1,\ldots, n$,
$\lambda_i$ is true at $\w_i$ and not true at $\w_j$ if
$j \neq i$.
Moreover, the formulas $\lambda_i$ are mutually exclusive and
exhaustive; exactly one is true at each state in $\Omega$.
Finally, let $\Lambda
= \{ \lambda_i: i \leq n\} \cup \{\bigwedge_{i\leq
%
n} \neg \lambda_i\}$ and define $c$ by taking $c(\lambda_i)
= \kappa(\w_i)$;
$c(\bigwedge_{i\leq n} \neg \lambda_i)$ can be defined arbitrarily.
Thus, $(\Lambda,c)$ is a contract such that $c(\w_i) = \kappa(\w_i)$
for all $i \le n$.

Now consider an arbitrary finite subset $\W'$ of $\W$.  Let $\W''$ be
a maximal subset of $\W'$ that contains at most one element of each
$\sim$-cell.  We can apply the construction above to $\W''$ to get a
$\w^*$-articulable contract $(\Lambda,c)$ such that $c(\w)
= \kappa(\w)$ for all $\w \in \W''$.
We claim that in fact $c(\w') = \kappa(\w')$ for all $\w' \in \W'$.
For if $\w' \in \W' - \W''$, then $\w' \sim \w$ for some
$\w \in \W''$.  By assumption, $\kappa(\w') = \kappa(\w)$, and
by construction, we must have $c(\w') = c(\w)$.  Thus $\kappa(\w') =
c(\w')$.  It follows that $c(\w') = \kappa(\w')$ for all
$\w' \in \W'$.  
\end{proof}

\commentout{
It is straightforward to show that an
  assignment $\kappa: \W \to \End_1{\oplus}\End_2$ can be 
implemented by a contract in state $\w^*$ if and only if $\kappa$
   is constant on all states within a cell (since two states in the same cell cannot
   be distinguished by formulas that the agents are aware of).  More 
   precisely, we must have  that 
\begin{equation}\label{eq1}
\mbox{   $\kappa^{-1}((d_1,d_2))$ is a union of cells for all $(d_1,d_2) \in
  \End_1 \oplus   \End_2$. }
\end{equation}
Thus, we can view $\kappa$ as a function from cells to $\End_1 \oplus
\End_2$, since it must act the same way on each state in a cell.}

\begin{lemma}
\label{lem:sameU}
If $\w,\w'\in \K_i(w^*)$, $w' \sim w$, and $\U$ satisfies A2, then
$U_{i,\w'} = U_{i,\w}$.
\end{lemma}

\begin{proof}

First notice that if $w,w' \in \K_i(w^*)$, then $\A_i(\w) = \A_i(\w') =\A_i(\w^*)$. Thus, by the definition of $\sim$, 
if $\w \sim \w'$, then
$$\{ \phi \in \L : (M,\w) \models \phi \land A_i\phi \} =
\{ \phi \in \L : (M,\w') \models \phi \land A_i\phi \},$$
implying that $\prop_\w^{\A_i} = \prop_{\w'}^{\A_i}$ and
$\con_\w^{\A_i} = \con_{\w'}^{\A_i}$.
The lemma follows immediately from  A2. 
\commentout{
So, let $w' \in \K_j(w^*)$, $j\neq i$. By way of contradiction, assume that $U_{i,\w'}(d) \neq U_{i,\w}(d)$, for some $d$.
First, since $\w \in \K_i(\w^*)$ we have that $\A_i(\w) = \A_i(\w^*)$, and $\w' \in K_j(\w^*)$ with $j \neq i$, so $\A_j(\w') = \A_j(\w^*) =\A_i(\w^*)$.

Second, our assumption of different utilities , by the contrapositive of A2, posits the existence of a property\footnote{Or, concept. Everything follows \emph{mutatis mutandis}.} such that, without loss of generality, $P \notin \prop_\w^{\A_i}(d)$ and $P \in \prop_{\w'}^{\A_i}$. 

 There are two cases:
 
 \begin{itemize}
 \item[(i)] If $P \in \A_i(\w) = \A_i(\w^*)$ then $(M,\w) \models \neg P(d)$ and $(M,\w^*) \models A_i \neg P(d)$ but $(M,\w') \not\models \neg P(d)$. This contradicts the claim that $\w \sim \w'$. 
 \item[(ii)] If $P \notin \A_i(\w^*)$, $(M,\w') \models \exists Y(A_i Y(d) \land \neg A_j Y(d))$ and $(M,\w^*) \models A_i \exists Y(A_i Y(d) \land \neg A_j Y(d))$. There are two sub-cases:
 	\begin{itemize}
	\item[(ii.a)] $(M,\w) \not\models \exists Y(A_i Y(d) \land \neg A_j Y(d))$. This contradicts that $\w \sim \w'$ immediately.
	\item[(ii.b)] $(M,\w) \models \exists Y(A_i Y(d) \land \neg A_j Y(d))$. Therefore, there exists a $P'$ such that $P' \in A_i(\w) = A_i(\w^*) = A_j(\w')$ but $P' \notin \A_j(\w)$. But then
$(M,\w') \models A_j P'(d)$, $(M,\w^*) \models A_j(A_j P'(d))$ and $(M,\w) \not\models A_j P'(d)$,
contradicting $\w \sim \w'$.
	\end{itemize}
 \end{itemize}
 }
\end{proof}

%
For each $\sim$-cell $[\w]$, 
\commentout{
choose an element $\w_{rep} \in [\w]$ such that
\begin{itemize}
\item $\w_{rep} \in \K_1(\w^*)$ if $[\w] \cap \K_1(\w^*) \neq \emptyset$;
\item $\w_{rep} \in \K_2(\w^*)$ if $[\w] \cap \K_1(\w^*) = \emptyset$ and
$[\w'] \cap \K_2(\w^*) \neq \emptyset$;
\item $\w_{rep}$ is an arbitrary element of $[\w]$ otherwise.
\end{itemize}
}
define $U_{i,[\w]}: D \to \R$ by taking $U_{i,[\w]} = U_{i,\w'}$ for some 
$\w' \in \K_i(\w^*) \cap [\w]$ if $\K_i(\w^*) \cap [\w]$
is nonempty, and define $U_{i,[\w]}$ arbitrarily otherwise.

By Lemma \ref{lem:sameU}, $U_{i,[\w]}$ is well defined (in that it
does not depend on the choice of $\w'$).  Define
%
%
%
%
\begin{equation}\label{eq3}
\begin{array}{ll}
\kappa(\w) \in \argmax_{(d_1,d_2) \in \End_1{\oplus}\End_2} 
 \big\{U_{1,[\w]}(d_1) + U_{2,[\w]}(d_2):  \\
U_{i,[\w]}(d_i) \geq  \max_{d \in \End_i} U_{i,[\w]}(d), i=1,2 \big\}.
\end{array}
\end{equation}
%


We assume that the elements of $\End_1 \oplus \End_2$ are ordered so
that if more than element in $\End_1 \oplus \End_2$ is the argmax in
(\ref{eq3}), we always choose the least one.  This guarantees that
$\kappa(\w) = \kappa(\w')$ if $\w \sim \w'$.
Thus,
by Lemma \ref{Lemma:implemented}, there exists a $\w^*$-articulable
contract $\hat c$ such that $\hat c(\w) =  \kappa(\w)$ for all
$\w \in \K_1(\w^*) \cup \K_2(\w^*)$.
The next lemma shows that this contract is acceptable and efficient, and thus
establishes the claim of part (a).

\commentout{
By Lemma \ref{Lemma:implemented} we can implement this over
$\K_1(\w^*) \cup \K_2(\w^*)$, using a contract $\hat c$ 
that will turn out to be efficient and acceptable. We first show that
it satisfies these properties at the representative worlds $\w_{rep}$,
then, using 
}

\begin{lemma}\label{lem:a} The contract $\hat c$ is $\w^*$-efficient and
  $\w^*$-acceptable for $i = 1,2$. 
\end{lemma} 

\begin{proof}

For $\w \in \K_i(\w^*)$, $U_{i,\w} = U_{i,[\w]}$, so $\w^*$-acceptability follows from
the constraint of the maximization problem give by (\ref{eq3}).

To see that $\hat{c}$ is $\w^*$-efficient,
suppose, by way of contradiction, that there exists 
$(d_1, d_2) \in \End_1{\oplus}\End_2$ such that  
$$U_{i,\w^*}(d_i)) \geq U_{i,\w^*}(\hat c_i(\w^*)) =
U_{i,\w^*}(\kappa_i(\w^*))$$ 
for $i =1,2$, with at least one inequality strict.
%
Since $\w^* \in \K_1(\w^*)\cap \K_2(\w^*)$, $U_{i,\w^*} = U_{i,[\w^*]}$
for $i = 1,2$.
It follows that $(d_1, d_2)$ satisfies the constraints given by
(\ref{eq3}). But, since one inequality is strict, we have
\begin{align*}
U_{1,[\w^*]}(d_1) + U_{2,[\w^*]}(d_2) > \\
U_{1,[\w^*]}( \kappa_1([\w^*])) + U_{2,[\w^*]}(( \kappa_2([\w^*]),
\end{align*}
contradicting the definition of $ \kappa$. Hence no such $(d_1, d_2)$
exists, and $\hat c$ is $\w^*$-efficient.
\end{proof}

\smallskip

Part (a) of Theorem~\ref{thm:contract1} follows from
Lemma~\ref{lem:a}.  For part (b), note that if $\A_i$ is constant on
$\K_1(\w^*) \cup \K_2(\w^*)$ for $i = 1,2$,
Lemma \ref{lem:sameU} can be strengthened to show that if 
$\w,\w' \in \K_1(\w^*) \cup \K_2(\w^*)$ then $U_{i,\w} = U_{i,\w'}$. 
The rest of the proof follows 
follows immediately from the added assumption that $\A_i$ is constant
on $\K_1(\w^*) \cup \K_2(\w^*)$ for $i=1,2$.
The rest of the proof follows
the proof of part (a) with the obvious adjustments.

For part (c), we consider a coarsening of $\sim$
that identifies states 
that agree on all sentences $\phi(d)$ with $\phi \in \Lbc$ that the
agents are aware of in $\w^*$.  More precisely,
$\w \sim^{bc} \w'$ if 
 \begin{align*}
\{\phi(d) : \phi \in \Lbc, (M,\w) \models \phi(d), (M,\w^*) \models
A_1\phi(d) \} = \\ 
\{\phi(d) : \phi \in \Lbc, (M,\w') \models \phi(d), (M,\w^*) \models
A_1\phi(d) \}. 
\end{align*}
The proof follows that of part (a) except the analogue of
Lemma \ref{Lemma:implemented} uses $\sim^{bc}$-cells and in
Lemma \ref{lem:sameU}, we appeal to A3 rather than A2. 
\end{proof}

\section{More on the axiom system}

Some comments on the axioms and rules of inference: AGP follows from
our assumption that awareness is generated from primitives: agent $i$
is aware of $\phi$ exactly if $i$ is aware of all the constant,
predicate, and concept symbols in $\phi$, and the assumption that
agents are aware of all the constant symbols in $\O$.  KA captures the fact that
agents know what they are aware of.  
T, 4, and 5 are the standard axioms of knowledge (when knowledge is
characterized by an equivalence relation), but 5
must be modified to deal with lack of awareness; the approach for doing
so goes back to Fagin and Halpern \citeyear{FH}.  A0 captures the fact
that to (explicitly) know $\phi$, $i$ must be aware of $\phi$.
%
$1_{\forall x}$, $1_{\forall X}$,
${\rm N}_{\forall x}$,
${\rm N}_{\forall X}$, ${\rm K}_{\forall x}$,
and ${\rm K}_{\forall X}$ are standard adaptations of axioms for
first-order logic, although we remark that the validity of $1_{\forall
  x}$ depends on the fact that the constants $d$ have the same
interpretation in all states (otherwise $\forall x K_i(P(x)) \rimp
K_i(P(c))$ would not be valid).

The Barcan$_x$ axiom is a standard axiom of first-order modal logic.
The converse is true as well (and provable from the other axioms).  
However, its validity depends on the fact that all the objects in $\O$
are in the language of each state.  As observed by HR, the analogue
for predicate variables does not hold.  For example, suppose that the
only predicate symbol in $\Phi(\w)$ is $P$ agent $i$ is aware of $P$
at $\w$.  It is easy to see that (since agents know what they
are aware of) $\forall X K_i A_i(X(c))$ holds at $\w$.  However,
if $i$ considers a world $\w'$ such that $\Phi(\w')$ includes $Q$ and
$i$ is not aware of $Q$ at $\w'$, then $K_i(\forall X A_i(X(c)))$ does
not hold at $\w$.  As discussed by HR, Barcan$_X$ is essentially the
closest approximation to the Barcan formula that we can get in models.
To understand it, suppose that $\phi$ is a formula with a free
predicate variable $X$.  Then Barcan$_X$ says that if, in a state
$\w$, agent $i$ is aware of all the predicates in $\phi$ and
if, no matter what formula $\psi \in\Lbc \cap \Phi(\w)$ is substituted
for $X$ in $\phi$, if $i$ is aware of all the predicates in $\psi$,
then $i$ knows $\phi[X/\psi]$, then in any state $\w'$ that $i$
considers possible where $i$ is aware of all predicates (so that
$\Phi(\w')$ just includes those predicates that $i$ is aware of),
$\forall X \phi$ holds.  It is not hard to see that this is sound.

In general, the axiom $\neg A_i (\phi)) \rimp K_i(\neg A_i \phi)$ is not valid;
If an agent is not aware of a sentence, then she can't know it, since
we are working with explicit knowledge (even though $\neg A_i \phi$ is
true in all the worlds that the agent considers possible.  FA$_X$ (which
essentially already appears in HR) gives a weak version of that axiom;
it says that if the agent is not aware of any predicates, then she
knows that.  The soundness of Fin$_x$ is clear.  Note that if $\O$ is
infinite, there is no corresponding axiom.

Although MP is standard, it is not obviously sound.  To understand the
problem, suppose
that $\phi$ and $\psi$ are sentences such that $\phi$ and $\phi \rimp
\psi$ are valid. To show 
that MP is sound, we would have to show that $\phi$ is valid. Suppose
not. Then there is a model $M^*$ and state $\w^*$ in $M^*$ such that
$(M^*, \w^*) \models \neg\psi$. 
If $(M^*, \w^*) \models \phi$ and $(M^*,\w^*)
\models \phi \rimp \psi$, then we have an immediate 
contradiction. Unfortunately, despite the validity of $\phi$ and $\phi \rimp
\psi$, we cannot conclude that $(M^*, \w^*) \models \phi$ and that $(M^*, \w^*)
\models \phi \rimp \psi$.  These sentences might not be in the language
at state $\w$.  The assumption that $\P$ is infinite allows us to deal
with this difficulty.

We now prove Theorem~\ref{thm:completeness}, which says that AX is a
sound and complete axiomatization if the set $\P$ of predicates is
infinite.  Since the proof follows exactly the
same lines as the HR soundness and completeness proof, we focus here
on the main differences.  

Soundness is straightforward, with the exception of MP,
Gen$_{\forall_x}$, Gen$_{\forall_X}$, 
and Barcan$_X$.
The soundness of MP is proved in Lemma A.1, Proposition A.1, and
Corollaries A.1 and A.2 in \cite{HR09}, under the assumption that the
set of primitive propositions is infinite;
the soundness of Barcan$_X$ is proved in Proposition A.3 in \cite{HR09}
(where the axiom is labeled Barcan$_X^*$). Essentially the identical
proofs show soundness in our setting.
We now prove that Gen$_{\forall_x}$ and Gen$_{\forall_X}$ are sound.

\begin{proposition} Gen$_{\forall_x}$ and Gen$_{\forall_X}$ are sound.
\end{proposition}

\begin{proof} 
For Gen$_{\forall_x}$, suppose that $\models \phi$.  We claim that we
must have $\models  \phi[c/c']$ for all 
constants $c'$.  For suppose not.  Then there must be some model $M$
and state $\w$ such that $(M,\w) \models \neg \phi[c/c']$.  
Let $M'$ be a model
just like $M$, but with the roles of $c$ and $c'$ reversed.
Specifically, if $I^{M}$ and $I^{M'}$ denote the interpretation
functions of $M$ and $M'$, respectively,%
\footnote{We use analogous notation consistently below.}
then for each predicate $P$ and
each state $\w$, $I^{M'}$
is the result of replacing $c$ in $P_\w$ by $c'$
and replacing $c'$ in $P_\w$ by $c$.
It is easy
to show, by induction on the structure of $\psi$, that, for all sentences
$\psi$, we have that $(M,\w) \models \psi[c/c']$ iff $(M',\w) \models
\psi$.
(In doing the induction, we first consider how many occurrences of
$\forall X$, that is, how many instances of quantification over
predicates, there are in 
$\psi$, and then do a subinduction on the length of $\psi$.)
It follows that $(M',\w) \models \neg \phi$, contradicting the
validity of $\phi$.  Thus, $\models \phi[c/c']$ for all constants
$c'$, so $\models \forall x \phi$, as desired.

For Gen$_{\forall_X}$, again suppose that $\models \phi$. We want to
show that $\models \forall X \phi[P/X]$. Suppose not. Then there must
be some model $M$, state $\w$ in $M$, 
and sentence $\psi \in  \Lbc(\Phi^M(\w))$ such that $(M,\w) \models
\neg \phi[P/\psi]$. %
Let $M'$ be 
identical to $M$ except that, for each state $\w'$, (a) $\Phi^{M'}(\w')$
is the resulting of adding $P$ to $\Phi^{M'}(\w')$ iff $\psi \in
\L(\Phi^{M}(\w'))$; (b) $I^{M'}$ is such that, for all $c \in \O$, we have
that $(M',\w') \models P(c)$ iff $(M,\w') \models \psi(c)$; and (c) the
awareness functions are such that, for all agents $i$, $(M',\w')
\models A_i(P(c))$ iff $(M,\w') \models A_i(\psi(c))$.  Intuitively,
$\psi$ plays the same role in $M$ that $P$ does in $M'$.

Then a
straightforward induction 
on the structure of sentences shows that for all states $\w'$ and all
sentences $\phi'$, we have
that $(M, \w') \models  \phi'[P/\psi]$ iff $(M', \w') \models \phi'$.
Since, by assumption, $(M,\w) \models \neg \phi[P/\psi]$, it follows
$(M',\w') \models \neg
\phi$, contradicting the validity of 
$\phi$. 
%
%
Hence, we have that $\forall X \phi[P/X]$, as desired. 
\end{proof}

For completeness, fix a
language $\L(\O,\P,\C)$ where $\P$ is infinite.
We now list the major lemmas used by HR to prove completeness,
expanding only on the differences.
As usual, the idea of the completeness proof is to construct a
canonical model $M^c$ where the worlds are maximal consistent
sets of sentences. We then show that if $\w_\Gamma$ is the world
corresponding to the maximal consistent set $\Gamma$, then
$(M^c,\w_\Gamma)\models 
\varphi$ iff $\phi\in \Gamma$.  As observed in 
\cite{HR05b}, this will not
quite work in the presence of quantification; 
there may be a maximal
consistent set $\Gamma$ of sentences such that $\neg\forall X\phi\in \Gamma$,
but $\phi[X/\psi] \in \Gamma$ for all $\psi\in \Lbc$. That is, there is
no witness to the falsity
of $\forall X \phi$ in $\Gamma$. This problem was dealt with in \cite{HR05b} by
restricting to maximal consistent sets $\Gamma$ that are
\emph{acceptable} in the sense that if $\neg \forall X \phi \in \Gamma$,
then $\neg \phi[X/q] \in \Gamma$ for infinitely many
primitive propositions 
$q \in \Phi$.  (Recall that HR consider a propositional language; 
this notion of acceptability requires the language to include
infinitely many primitive propositions, as HR assumed it did.)
As in \cite{HR09}, because here we have possibly different languages
associated different worlds, we need to consider acceptability and
maximality with respect to a language.  The following definition is
adapted from \cite{HR09}.

\begin{definition}
A sentence $\phi$ is \emph{provable} from AX,
denoted $\mbox{AX} \vdash \phi$, if there is a sequence of sentences such
that the last one is $\phi$, and each one is either an instance of an
axiom of AX  or 
follows from previous sentences in the sequence by an application
of an inference  rule.  We write $\Gamma \vdash \phi$ if there are
sentences $\beta_1, \ldots, \beta_m$ in $\Gamma$ such that AX
$\vdash (\beta_1 
\land \ldots \land \beta_m) \rimp \phi$.
\end{definition}

\begin{definition}
A set $\Gamma$ of sentences is {\em acceptable with respect
to the language $\L(\O,\P',\C')$}
if, for all sentences $\phi\in \L(\O,\P',\C')$, if
$\Gamma\vdash \phi[X/P]$ for all 
but finitely predicate symbols  $P \in \P'$, then $\Gamma\vdash
\forall X \phi$.
\end{definition}

\begin{definition}
$\Gamma$ is a {\em maximal AX-consistent set of sentences with
respect to $\L(\O,\P',\C') \subseteq \L(\O,\P,\C)$}
if $\Gamma$ is an AX-consistent subset of $\L(\O,\P',\C')$ and, for
all sentences $\phi\in \L(\O,\P',\C')$, 
if $\Gamma\cup\{\phi\}$
is $AX$-consistent, then $\phi\in\Gamma$.  
\end{definition}

The following four lemmas are essentially Lemmas A.4, A.5, A.6, and A.7 in
\cite{HR05b}.  Since the proofs are essentially identical, we do not repeat them
here.

\begin{lemma}
\label{Claim4K}
If $\Gamma$ is a finite
set of sentences,
then $\Gamma$ is acceptable with
respect to every language $L(\O,\P',\C')$ such that $\P'$ is infinite.
\end{lemma}

\begin{lemma}
\label{Claim5K}
If $\Gamma$ is
acceptable with respect to
$\L(\O,\P',\C')$ and $\psi$ is a sentence
in $\L(\O,\P',\C')$,
then $\Gamma \cup \{\psi\}$ is
acceptable with respect to
$\L(\O,\P',\C')$. 
\end{lemma}

\begin{lemma}
\label{Claim6X}
If $\Gamma\subseteq \L(\O,\P',\C')$
is an acceptable AX-consistent
set of sentences with respect to $\L(\O,\P',\C')$,
then $\Gamma$ can be extended to a set of sentences that is
acceptable and  maximal
AX-consistent with respect to $\L(\O',\P',\C')$.
\end{lemma}

Let $\Gamma/K_i=\{\phi:K_i\phi\in\Gamma\}$.
\begin{lemma}
\label{LemmaA3X}
If $\Gamma$ is 
a maximal $AX$-consistent set of sentences with respect to $\L(\O,\P',\C')$
that contains $\neg
K_i\phi$ and $A_i\phi$, then $\Gamma/K_i\cup\{\neg\phi\}$ is
$AX$-consistent.
\end{lemma}

The next lemma is essentially Lemma A.6 in \cite{HR09}, and again, its
proof is almost identical.  We say $\P'$
is \emph{co-infinite} if $\P - \P'$ is infinite.

\begin{lemma}\label{LemmaA4X}
If $\Gamma$ is
an acceptable
maximal AX-consistent set of
sentences with respect to $\L(\O,\P',\C')$, where
$\P'$ is infinite and co-infinite,
$\neg K_i\phi\in \Gamma$, and $A_i\phi\in\Gamma$, then there
is an infinite and co-infinite set $\P''$ of predicates and a
set $\Delta$ of sentences
that is an acceptable,
maximal AX-consistent set with respect to $\L(\O,\P'',\C')$ and
contains
$\Gamma/K_i \cup \{\neg\phi\}$.
Moreover, $A_i\psi \in \Delta$ iff $A_i\psi\in\Gamma$ for
all sentences $\psi$.
\end{lemma}

The following lemma is the one new lemma that we need to deal with
concepts and variables.  It illustrates the role of
acceptability in the construction.

\begin{lemma}\label{lem:new}
If $\Gamma$ is
an acceptable maximal AX-consistent set of
sentences with respect to $\L(\O,\P',\C')$ and $C \in \C'$,
then there exists a sentence
$\psi \in \Lbc(\P')$ such that $\forall x
(C(x) \dimp \psi(x)) \in \Gamma$.
\commentout{
\item[(b)] If $\Gamma$ is
an object-acceptable maximal AX-consistent set of
sentences with respect to $\L(\O',\P',\C')$, then for all object
variables $x$, there exists a constant $c \in \O'$ such that $\forall
X (X(x) \dimp X(c)) \in \Gamma$ and for all sentences $\phi(x)$,
$\phi(x) \dimp \phi(c) \in \Gamma$.
\end{itemize}
}
\end{lemma}


\begin{proof} 
Suppose, by way of contradiction, that there is no
sentence $\psi \in \Lbc(\P')$ such that $\forall x
(C(x) \iff \psi(x)) \in \Gamma$.  Then, since $\Gamma$ is maximal,
for all predicates $Q \in \P'$, we must have that
$\neg \forall x (C(x) \iff Q(x)) \in \Gamma$.  Since $\Gamma$ is
acceptable, it follows that $\forall X \neg \forall x (C(x) \iff
X(x)) \in \Gamma$.  But since $\Gamma$ contains
every instance of axiom Con, it contains $\exists X  
(\forall x (C(x) \iff X(x)))$ (i.e., $\neg \forall X \neg 
(\forall x (C(x) \iff X(x))) \in \Gamma$).  Thus,
$\Gamma$ is inconsistent, a contradiction.
\end{proof}

We now complete the completeness proof by constructing a canonical
model, essentially as is done in \cite{HR09}.

\begin{lemma}
\label{LemmaA5X}
If $\varphi$ is an $AX$-consistent sentence, then $\varphi$ is satisfiable.
\end{lemma}

\begin{proof}
We construct a canonical model where the worlds are maximal
consistent sets of sentences.  However, now the worlds must also
explicitly include the language, and we assume 
that the language associated with each world is infinite and co-infinite.
Let $M^{c}=
(\W^c,\OO, \Phi^c,\A_1^c, \ldots, \A_n^c, \K_1^c,\ldots,\K_n^c,I^c)$ 
be the canonical awareness structure for $\L(\O,\P,\C)$,
constructed  as follows:
\begin{itemize}
\item $\W^c =\{(\w_\Gamma,\L(\O,\P',\C')): \Gamma$ is
a set of sentences that is acceptable and maximal
AX-consistent with respect to $\L(\O,\P',\C')$, 
$L(\O,\P',\C') \subseteq \L(\O,\P,\C)$, and $\P'$
is infinite and co-infinite$\}$;

\item $\OO = \O$;

\item $\Phi^c((\w_\Gamma,\L))=\L$;

%
\item $\A_i^c\w_\Gamma,\L) = \{\psi: A_i \psi \in \Gamma\}$;

\item $\K_i^c((\w_\Gamma,\L))= \{(\w_\Delta,\L'):\Gamma/X_i\subseteq
\Delta \mbox{ and } A_i\phi\in \Gamma \mbox{ iff }A_i\phi\in 
\Delta \mbox{ for all sentences }\phi\}$;
\item $I^c_{(\w_\Gamma,\L)}(P) = \{c: P(c) \in \Gamma\}$;
\item $I^c_{(\w_\Gamma,\L)}(C) = \psi$ if $\forall
x(C(x) \dimp \psi(x)) \in \Gamma$.  (By Lemma~\ref{lem:new}, there
will be such a $\psi$; if there is one more than one, we can pick the
least one in some ordering of formulas.) 
\end{itemize}

Standard arguments now show that $\psi \in \Gamma$ iff $(M^c,
(\w_\Gamma, \L)) \models \psi$. It 
follows from Lemmas A.1 and A.3 that  every
consistent sentence is in some acceptable and maximal consistent set
with respect to a language $\L(O, \P', \C')$ such that $\P'$ is infinite
and co-infinite. Thus, if $\phi$ is consistent, there is some state $(\w_\Gamma,\L)$
in the canonical model such that $(M^c, (\w_\Gamma, \L)) \models
\phi$. This proves the lemma and the theorem. 
\end{proof}
}
\end{document}